\documentclass[letterpaper,conference]{ieeeconf}

\usepackage{cite}
\usepackage{amsmath,amssymb,amsfonts}
\usepackage{algorithmic}
\usepackage{booktabs}
\usepackage{graphicx}
\usepackage{textcomp}
\usepackage{xcolor}
\newtheorem{theorem}{Theorem}

\title{Layered construction of Message-Wise Unequal Error Protection Codes}

\author{Qiming Lu, Shan Lu, Takaya Yamazato,\\
Dept. of Information and Communication Engineering, Nagoya University, Nagoya 464-8063, Japan\\
lu.qiming.h0@s.mail.nagoya-u.ac.jp,~shan.lu.jp@ieee.org,~yamazato@ieee.org
}

\begin{document}

\maketitle

\begin{abstract}
Conventional communication systems are mainly designed to reduce error rates and increase transmission rates, and therefore usually provide uniform protection to all transmitted messages. However, in intent-oriented applications, different messages may have different semantic meanings and importance levels, requiring different levels of reliability. This paper proposes a layered construction of message-wise unequal error protection (UEP) codes for short-blocklength communication. Instead of appending an explicit protection tag to each codeword, the proposed method embeds the protection structure directly into the Hamming-distance structure of the codebook. By assigning larger minimum intra-level distances to higher-importance message groups and imposing suitable inter-level distance constraints, the proposed codebook provides differentiated error-correction capabilities while enabling reliable importance-level classification at the receiver. Theoretical conditions for correct group classification are derived, and simulations over AWGN and VLC-ISI channels show that the proposed scheme improves BER performance and group classification accuracy compared with a tag-based ECC baseline.
\end{abstract}

\vspace{-5mm}
\section{Introduction}
With the rapid development of intelligent driving technologies, vehicle-to-everything (V2X) communication has become a key technology for intelligent transportation systems and autonomous driving. By enabling information exchange among vehicles, infrastructure, and other road users, V2X communication improves driving safety and traffic efficiency. In addition to conventional radio-frequency solutions, visible light communication (VLC) has been considered a promising complementary technology~\cite{yamazato,Asaoka2026}, owing to its low cost and seamless integration with vehicular lighting systems.

In such scenarios, intent-oriented communication has attracted increasing attention, where messages have heterogeneous importance levels. However, conventional systems typically provide uniform error protection, which is insufficient for such requirements when safety-critical and non-critical messages coexist. This motivates the use of message-wise unequal error protection (UEP) to provide differentiated reliability under short-blocklength constraints.

UEP has been studied for more than two decades~\cite{699346,7848728,7334457}. Most early studies focused on bit-wise UEP, where unequal protection is applied to specific symbols or bit positions within a single codeword. For example, existing work~\cite{6978587} has demonstrated the hardware feasibility of bit-wise UEP.
However, such approaches remain limited to symbol-wise protection inside individual codewords. More recently, several studies~\cite{5319734,9525108} have extended the concept of UEP to the message-wise, where different messages are assigned different reliability requirements.

The message-wise UEP is different from rate-compatible coding~\cite{Hagenauer1988,Hou2014,Hou2016}. 
Rate-compatible codes obtain multiple code rates from a mother code through puncturing, shortening, or incremental redundancy, and are mainly used for channel-adaptive rate control. 
For a selected rate, all messages in the transmitted block are protected by the corresponding coding structure, and no explicit separation among message groups with different protection capabilities is required. 
In contrast, the proposed message-wise UEP maps messages with different priority levels to different codeword groups and explicitly designs both intra-group and inter-group Hamming distances. Thus, the unequal protection capability is embedded in the geometry of the codebook itself, rather than being realized by selecting different rates from a rate-compatible code family.

Existing message-wise UEP schemes have focused on long-blocklength designs based on LDPC codes, which can achieve near-capacity performance and support theoretical performance analysis~\cite{5715839,10571004}. However, in low-speed VLC systems, the high encoding and decoding complexity of LDPC codes may introduce significant latency, making them less suitable for low-latency VLC scenarios. In contrast, short block codes provide an attractive alternative because of their hardware simplicity, low latency, and multi-bit error-correction capability, which are also consistent with VLC reference standards.

Tag-based UEP schemes have also been investigated~\cite{mark-UEP}, where a protection-level indicator is appended to each codeword and combined with an error-correcting code. However, such approaches have a key weakness: If the tag is decoded incorrectly, the receiver may use the wrong decoding rule. This problem is more serious in short-blocklength scenarios, because the tag adds extra bits and its reliability directly affects the overall performance.

To address these limitations, this paper proposes a layered construction method for message-wise unequal error protection codes for short-blocklength intent-oriented communication. 
By assigning different minimum intra-level distances to messages with different importance levels and imposing suitable inter-level distance constraints, the proposed codebook realizes differentiated error-correction capability while preserving reliable importance-level classification at the receiver. We further derive theoretical feasibility conditions for correct group classification under bounded errors and validate the proposed design through simulations over both additive white Gaussian noise (AWGN) and VLC inter-symbol interference (VLC-ISI) channels. The results show that the proposed tag-free layered UEP scheme achieves lower bit error rates and higher grouping accuracy than the conventional indicator-plus-ECC approach, especially under short-blocklength and interference-dominated VLC scenarios.


\vspace{-3mm}
\section{System construction}

In this section, we describe the design requirements for the proposed layered UEP codebook. This study focuses on command-level intent transmission\cite{9679803}, where each command is represented as a discrete message and assigned an importance level. The codebook is designed according to two principles. First, the minimum Hamming distance within each importance level determines the error-correction capability of that level. Second, the Hamming distance between different levels must be sufficiently large to enable reliable importance-level classification at the receiver. Therefore, higher-importance message groups are assigned larger intra-level minimum distances, while different groups are separated by suitable inter-level distances.

\subsection{Definition of m-level unequal error protection}\label{AA}
First, we construct a message set $\mathcal{M}$ that contains all the messages to be transmitted:
\begin{equation}
    \mathcal{M} =\bigcup_{x=1}^{m}\mathcal{M}_{x}, \hspace{1em}\mathcal{M}_{x}\bigcap \mathcal{M}_{y}=\oslash ,x \ne y 
\end{equation}
The class index $x$ represents the protection priority, where a larger $x$ corresponds to a higher protection level.Each message group $\mathcal{M}_{x}$ is mapped to a binary codeword group $C_{x}\subseteq \left \{ 0,1 \right \}^{n}$ through an injective encoder:
\begin{equation}
    f_{x}: \mathcal{M}_{x}\to C_{x} 
\end{equation}
where $\left | C_{x} \right | = N_{x}$ and $n$ is the codeword length. The overall layered codebook is therefore given by
\begin{equation}
    C=\left \{ C_{1},C_{2},\cdots ,C_{m} \right \} 
\end{equation}

An $m$-level unequal error protection  code is represented by the following formula:
\[
C_x = \left\{ c^{(x)}_1, c^{(x)}_2, \ldots, c^{(x)}_{N_x} \right\},
\qquad
c^{(x)}_i \in \{0,1\}^n .
\]
Here, \(c^{(x)}_i\) denotes the \(i\)-th codeword in group \(C_x\), and \(N_x\) is the number of codewords in that group. 

\subsubsection{minimum intra-group Hamming distance} For each group $C_{x}$, we define the \textit{minimum intra-group Hamming distance} as:
\begin{align}
    \begin{split}
       d_{\min}(C_x) = \min_{\substack{u,v \in C_x \\ u \neq v}} d(u,v)
    \end{split}
\end{align}
where $d(\cdot ,\cdot)$ denotes the Hamming distance.

The guaranteed error-correction capability of group \(C_{x}\) is defined as
\[
t_x = \left\lfloor \frac{d_{\min}(C_x)-1}{2} \right\rfloor .
\]
Equivalently, to guarantee correction of up to \(t_x\) errors, the intra-level minimum distance should satisfy
\[
d_{\min}(C_x) \geq 2t_x + 1 .
\]
Hence, a larger minimum intra-group distance implies a stronger protection capability for that group.

To realize message-wise unequal error protection,the groups are designed so that $t_{1} \le t_{2} \le \dots \le t_{m}.
    \label{eq2}$
That is, higher-priority message classes are assigned stronger protection.

\subsubsection{inter-level distance} 
The intra-level minimum distance \(d_{\min}(C_x)\) determines the error-correction capability within each importance level. However, for the proposed \textit{tag-free scheme}, the receiver must also be able to identify the importance level of a received vector without relying on an explicit protection indicator. Therefore, sufficient Hamming-distance separation is required between different codeword groups.

For two different groups \(C_p\) and \(C_q\), the inter-level distance is defined as
\begin{equation}
d_{pq} = \min_{\substack{u\in C_p, ~v\in C_q}}
d(u,v),
\qquad p\neq q .
\label{def:interleveldis}
\end{equation}
This value represents the minimum Hamming distance between any codeword in \(C_p\) and any codeword in \(C_q\). A larger inter-level distance reduces the probability that a received vector is incorrectly classified into another importance level.

To enable reliable nearest-group classification under bounded errors, we impose the following inter-level distance constraint:
\begin{equation}
d_{pq} \geq 2\max\{t_p,t_q\}+1,
\qquad \forall p\neq q .
\end{equation}
This condition is used as a design requirement for the layered codebook. Intuitively, it ensures that different importance levels are sufficiently separated in the Hamming-distance space. The theoretical justification of this condition is provided in Section~III. A stronger but convenient construction rule is
\begin{equation}
d_{pq} \geq 
\max\left\{
d_{\min}(C_p), d_{\min}(C_q)
\right\},
\qquad \forall p\neq q .
\end{equation}
Since \(d_{\min}(C_x)\geq 2t_x+1\), this stronger rule directly implies the previous inter-level distance condition. In this paper, this rule is adopted for the proof-of-concept codebook construction.

Beyond the minimum separation requirement, the inter-level distance can also be used to reflect the difference in message importance. Specifically, when possible, message groups with larger importance differences are assigned larger inter-level distances. In this way, the inter-level design serves two purposes: it supports reliable importance-level classification at the receiver and preserves the priority structure of the message set in the Hamming-distance space.



\subsection{Decoding}
Assume that a codeword \(c_i^{(A)}\in C_A\) is transmitted. 
After hard decision at the receiver, the received vector can be written as
\begin{equation}
    r = c_i^{(A)} \oplus e,
\end{equation}
where \(e\in\{0,1\}^n\) denotes the binary error vector and \(\oplus\) represents bitwise modulo-2 addition. The number of bit errors is given by
\begin{equation}
    w = \mathrm{wt}(e) = d(r,c_i^{(A)}),
\end{equation}
where \(\mathrm{wt}(\cdot)\) denotes the Hamming weight.

The receiver applies a nearest-group decoding rule as shown in Fig.~\ref{fig:Decoding process}. 
For each group \(C_x\), the minimum Hamming distance between the received vector \(r\) and the codewords in \(C_x\) is computed as
\begin{equation}
    d_x^{\min} = \min_{c\in C_x} d(r,c).
    \label{eq:group_min_distance}
\end{equation}
The estimated importance level is then determined by selecting the group with the smallest distance:
\begin{equation}
    \hat{s}= \arg \min_{1\leq x \leq m} d_x^{\min}.
    \label{eq:nearest_group_decoding}
\end{equation}
Here, \(\hat{s}\) denotes the estimated importance level of the received vector. If multiple groups have the same minimum distance, a predetermined tie-breaking rule can be used. Under the distance conditions discussed in Section~III, such ambiguity does not occur when the number of errors is within the correction capability of the transmitted group.

After the importance level is identified, the receiver estimates the transmitted codeword within the selected group as
\begin{equation}
    \hat{c}= \arg \min_{c\in C_{\hat{s}}} d(r,c).
\end{equation}
The decoded message is then obtained by applying the inverse mapping of the corresponding encoder:
\begin{equation}
    \hat{m}=f_{\hat{s}}^{-1}(\hat{c}).
\end{equation}

In this proof-of-concept study, the nearest-group decoder is implemented by \textit{exhaustive search} over the codebook. Although this approach directly verifies the feasibility of the proposed layered construction, its complexity increases with the total number of codewords. Therefore, the development of structured codebooks and low-complexity decoding algorithms remains an important direction for future work.
    \vspace{-4mm}
\begin{figure}[!h]
    \centering
    \includegraphics[width=\linewidth]{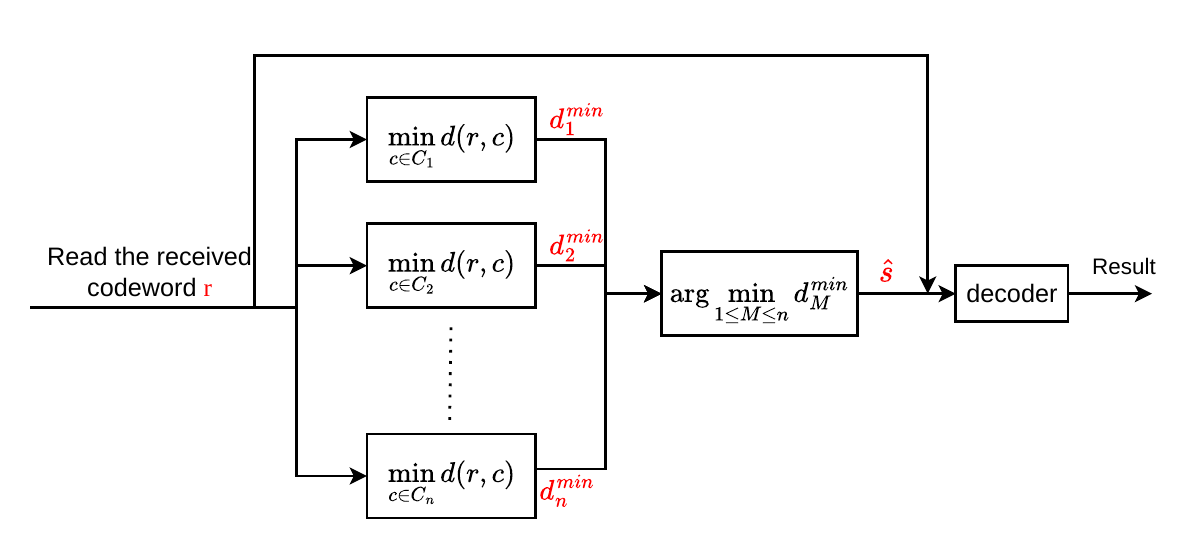}
    \caption{Decoding process}
        \vspace{-3mm}
    \label{fig:Decoding process}
\end{figure}

\vspace{-3mm}

\section{Theoretical Analysis of Group Classification}
This section provides a theoretical analysis of the proposed layered UEP codebook. 
We show that, under the inter-level distance constraint introduced in Section~II-B, the nearest-group decoder can correctly identify the transmitted importance level when the number of bit errors is within the correction capability of that level. 
We also discuss the case where the error weight exceeds this correction capability and characterize which groups cannot be falsely selected.

Since the proposed scheme does not use an explicit protection tag, the receiver must infer the importance level directly from the received vector. 
A key question is whether channel-induced bit errors may move the received vector closer to another group, leading to an incorrect importance-level decision. 
The following theorem addresses this issue by showing that, under a sufficient inter-level distance condition, any received vector whose error weight is within the correction capability of the transmitted group remains closer to that group than to all other groups. 
Therefore, the nearest-group decoder can correctly identify the original importance level under bounded bit errors.

\begin{theorem}
Assume that a codeword \(c_i^{(A)} \in C_A\) is transmitted and that \(r\) is the corresponding hard-decision received vector. Let
$w = d(r,c_i^{(A)})$
denote the number of bit errors. If
$$w \leq t_A
\quad \text{and} \quad
d_{AB} \geq 2t_A+1, \quad \forall B \neq A,$$
then the nearest-group decoder correctly identifies the transmitted group, i.e., $$\hat{s}=A.$$
\end{theorem}
\begin{proof}
Based on the nearest-group decoder, $d_x^{\min}$ is computed for each group \(C_x\) based on (\ref{eq:group_min_distance}).

First, consider the transmitted group \(C_A\). Since  \(c_i^{(A)} \in C_A\), the minimum distance from \(r\) to group \(C_A\) cannot be larger than the distance from \(r\) to \(c_i^{(A)}\). Therefore,
$$d_A^{\min}=\min_{c\in C_A} d(r,c)\leq d(r,c_i^{(A)})=w.$$
Since \(w\leq t_A\), we have 
\begin{equation}
d_A^{\min}\leq t_A .
\label{daminleata}
\end{equation}

Next, consider any other group \(C_B\), where \(B\neq A\). 
For any codeword \(c_j^{(B)}\in C_B\), the triangle inequality of the Hamming distance gives $$d(c_i^{(A)},c_j^{(B)})\leq d(c_i^{(A)},r)+d(r,c_j^{(B)}).$$
Because $d(c_i^{(A)},r)=w,$ we obtain
\begin{equation}
d(r,c_j^{(B)})\geq d(c_i^{(A)},c_j^{(B)})-w.
\label{qu:triangleGB}
\end{equation}

By the definition of the inter-level distance of (\ref{def:interleveldis}), any codeword in \(C_A\) and any codeword in \(C_B\) must satisfy $d(c_i^{(A)},c_j^{(B)})\geq d_{AB}.$
Using the assumed inter-level distance condition
$d_{AB}\geq 2t_A+1$ and \(w\leq t_A\),
we further obtain
\[
d(r,c_j^{(B)})
\geq
d_{AB}-w
\geq
2t_A+1-w \geq t_A+1.
\]
Therefore,
\[
d(r,c_j^{(B)})\geq t_A+1.
\]

This inequality holds for every \(c_j^{(B)}\in C_B\). Hence, the minimum distance from \(r\) to group \(C_B\) also satisfies
\begin{equation}
d_B^{\min} = \min_{c\in C_B}d(r,c) \geq t_A+1.
\label{daminta}
\end{equation}

Combining the results (\ref{daminleata}) and (\ref{daminta}), we have $
d_A^{\min}<d_B^{\min},
\qquad \forall B\neq A.$ Therefore, the received vector \(r\) is strictly closer to the transmitted group \(C_A\) than to any other group based on (\ref{eq:nearest_group_decoding}). 
Consequently, the nearest-group decoder correctly identifies the transmitted importance level:
$\hat{s}=A.$
This completes the proof.
\end{proof}

Theorem 1 shows that inter-level distance condition provides a sufficient guarantee for correct importance-level classification under bounded bit errors. Specifically, when the error weight does not exceed the correction capability of the transmitted group, the received vector remains strictly closer to the transmitted group than to any other group. 
This result confirms that the proposed tag-free codebook can support reliable group identification without relying on an explicit protection indicator.

When the error weight exceeds the correction capability of the transmitted group, correct classification is no longer guaranteed. 
However, it is still important to characterize which groups cannot be falsely selected by the nearest-group decoder. The following theorem analyzes this case and shows that \textit{groups with sufficiently large correction radii remain protected from false selection}.
\begin{theorem}
Assume that a codeword \(c_i^{(A)}\in C_A\) is transmitted and that \(r\) is the corresponding hard-decision received vector and the number of bit errors
$w>t_A$. 

For any group \(C_B\), \(B\neq A\), if
$$d_{AB}\geq 2t_B+1
\quad \text{and} \quad
t_B\geq w,
$$
then \(C_B\) cannot be selected by the nearest-group decoder. Consequently, if the nearest-group decoder misclassifies the received vector, the selected group \(C_{\hat{s}}\) must satisfy
$
t_{\hat{s}}<w .
$
\end{theorem}
\begin{proof}
Consider any group \(C_B\), \(B\neq A\), satisfying \(t_B\geq w\). 
Since \(c_i^{(A)}\in C_A\), the minimum distance from \(r\) to the transmitted group satisfies
\[
d_A^{\min}
=
\min_{c\in C_A}d(r,c)
\leq d(r,c_i^{(A)})
=
w .
\]

For any codeword \(c_j^{(B)}\in C_B\), using the same triangle-inequality argument as in the proof of Theorem~1, we have
$
d(r,c_j^{(B)})
\geq
d(c_i^{(A)},c_j^{(B)})-w .
$

By the definition of the inter-level distance,
$d(c_i^{(A)},c_j^{(B)})\geq d_{AB}.$
Together with the assumed condition \(d_{AB}\geq 2t_B+1\), this gives
$
d(r,c_j^{(B)})
\geq
d_{AB}-w
\geq
2t_B+1-w .
$
Since \(t_B\geq w\), we have
$
2t_B+1-w \geq w+1.
$
Therefore,
$
d(r,c_j^{(B)})\geq w+1 .
$

Since this holds for every \(c_j^{(B)}\in C_B\), the minimum distance from \(r\) to group \(C_B\) satisfies
\[
d_B^{\min}
=
\min_{c\in C_B}d(r,c)
\geq w+1 .
\]
Combining this with \(d_A^{\min}\leq w\), we obtain
$
d_A^{\min}<d_B^{\min}.
$
Thus, \(C_B\) cannot be selected by the nearest-group decoder.

Since the above argument applies to any group \(C_B\) with \(t_B\geq w\), no such group can be falsely selected. Consequently, if misclassification occurs, the selected group \(C_{\hat{s}}\) must satisfy
$
t_{\hat{s}}<w .
$
This completes the proof.
\end{proof}

Together, Theorem~1 and Theorem~2 justify the proposed inter-level distance design. Theorem~1 guarantees correct importance-level classification under bounded bit errors, while Theorem~2 limits the possible false-selection behavior when the error weight exceeds the correction capability of the transmitted group. These results provide the theoretical foundation for the proposed tag-free layered UEP codebook.

\section{Simulation and analysis}
In this proof-of-concept study, the layered codebook is obtained by exhaustive search over candidate binary codewords, where a candidate is accepted only if all prescribed intra-level and inter-level Hamming-distance constraints are satisfied.

We evaluate the proposed scheme and the baseline over additive white Gaussian noise (AWGN) and visible light communication inter-symbol interference (VLC-ISI) channels\cite{Asaoka2026, Asaoka2025, Shi2025}. 
The VLC-ISI channel is modeled as a binary-input channel with adjacent-symbol interference. In this model, the received intensity of the current bit depends not only on the current transmitted bit but also on the optical leakage from neighboring bits. Following the simplified VLC-ISI model in~\cite{Asaoka2026, Shi2025}, the received intensity is expressed as
$$
y(t)=h x(t-1)+x(t)+h x(t+1), \qquad x(t)\in\{0,1\},$$
where \(h\) denotes the interference coefficient and \(y(t)\) represents the received luminance intensity before hard decision.

For comparison, we use a tag-based BCH scheme as the baseline~\cite{mark-UEP}.
For fairness, both the proposed scheme and the tag-based BCH baseline use the same total codeword length of 45 bits, the same number of importance levels, and the same group sizes. In the baseline, each codeword consists of a protection-level indicator and an error-correcting codeword. The indicator specifies the error-correction capability used for decoding. Under a fixed total codeword length of 45 bits, each baseline codeword is composed of a 14-bit indicator concatenated with a 31-bit BCH codeword. 

To evaluate different priority requirements, we define six importance levels, denoted by Level A to Level F in increasing order of importance. For clarity in the simulation figures, three representative levels, namely Level A, Level D, and Level F, are selected for comparison.

\vspace{-3mm}
\subsection{BER performance}
The BER performance over the AWGN and VLC-ISI channels is shown in Fig.~\ref{fig2} and Fig.~\ref{fig3}, respectively. For visual clarity, three representative importance levels, namely Levels A, D, and F, are selected for comparison.

Fig.~\ref{fig2} shows the BER performance over the AWGN channel. The proposed scheme consistently outperforms the tag-based "indicator + ECC" baseline over the entire SNR range. At a BER of \(10^{-4}\), the proposed scheme achieves an SNR gain of approximately \(4\) dB for Group F and approximately \(3\) dB for Group D compared with the baseline.

Fig.~\ref{fig3} shows the BER performance over the VLC-ISI channel. As the interference coefficient \(h\) increases, the proposed scheme maintains a lower BER than the tag-based baseline, demonstrating stronger robustness against inter-symbol interference. Although the performance gap between the high-priority groups is relatively small, the proposed scheme provides a clear improvement for the lower-priority groups.

\begin{figure}[!h]
    \centering
    \includegraphics[width=0.85\linewidth]{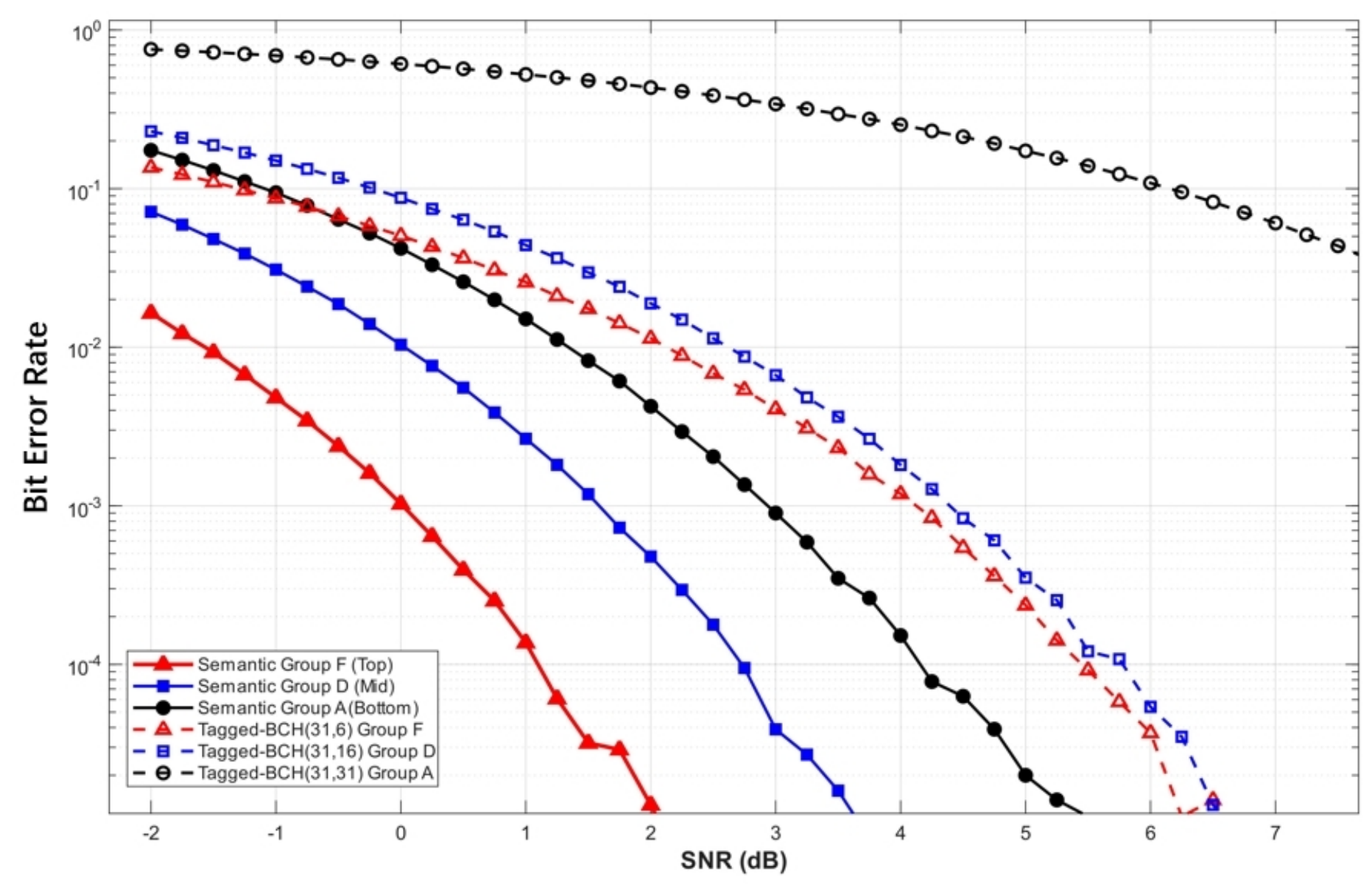}
        \vspace{-3mm}
    \caption{BER of AWGN Channel} 
    \label{fig2}
    \includegraphics[width=0.85\linewidth]{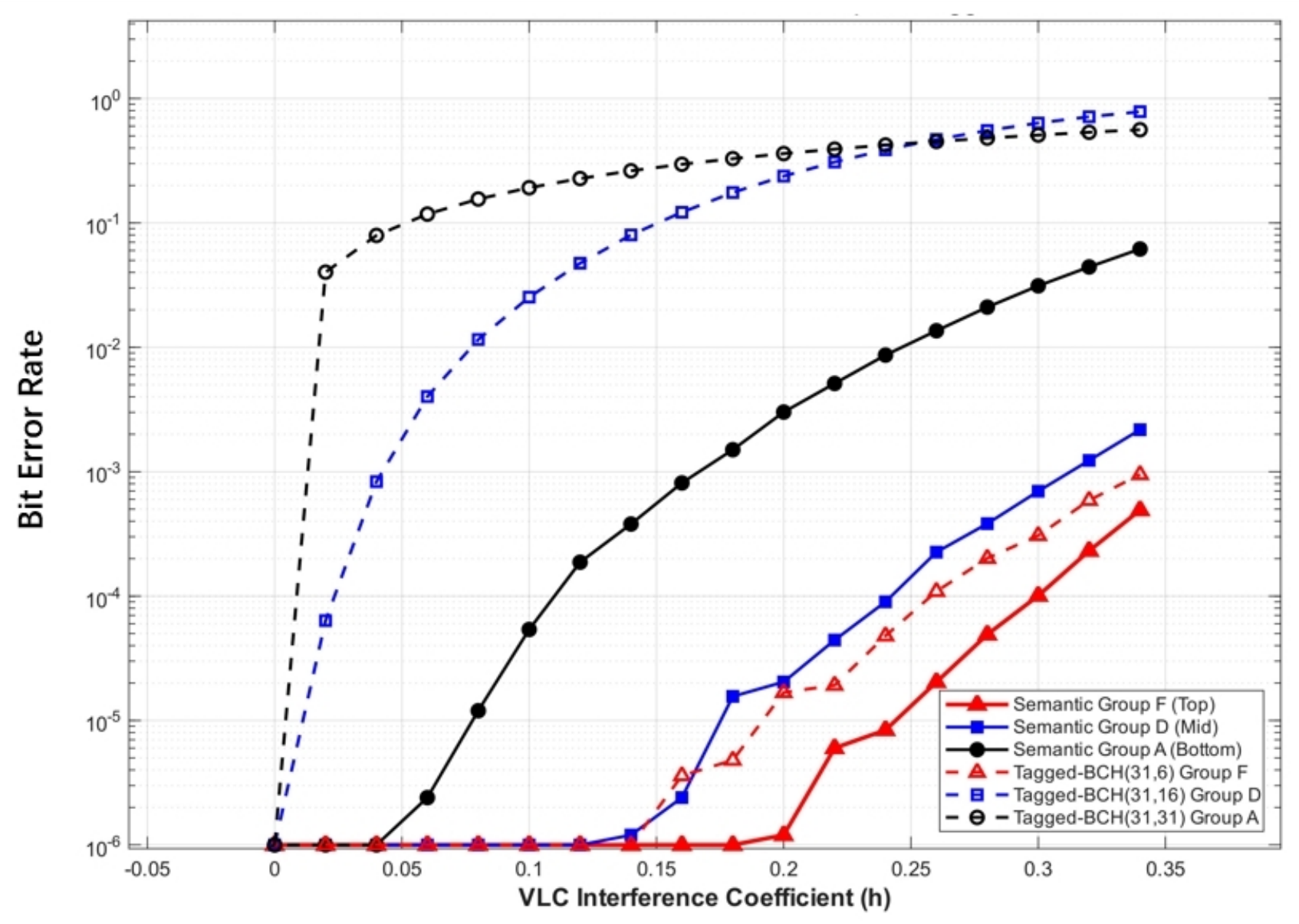}
        \vspace{-3mm}
    \caption{BER of VLC Channel} 
    \label{fig3}
    \vspace{-5mm}
\end{figure}

\subsection{Group Classification Accuracy}
In addition to the BER, we evaluate the reliability of importance-level classification using Group Classification Accuracy (GCA). GCA is defined as the probability that the receiver correctly assigns a received codeword to its original importance level:
$\mathrm{GCA} = \Pr(\hat{s}=s),$ where \(s\) denotes the transmitted importance level and \(\hat{s}\) denotes the level estimated by the receiver.

In the proposed layered UEP scheme, correct importance-level classification is essential because the receiver applies the corresponding error-correction capability according to the estimated level. Therefore, GCA measures whether the proposed Hamming-distance-based codebook can preserve the intended priority structure under channel noise. For a fair comparison, the baseline is configured with the same number of importance levels, group sizes, and error-correction capabilities as the proposed scheme.

The GCA performance over the AWGN and VLC-ISI channels is shown in Fig.~\ref{fig5} and Fig.~\ref{fig6}, respectively. In the AWGN channel, both the proposed scheme and the tag-based baseline achieve high classification accuracy at large SNR values. However, in the low-SNR region, the proposed scheme provides better classification performance than the tag-based approach. The improvement is especially clear for the highest-importance level, where the performance gap reaches approximately \(10\%\).

In the VLC channel, the GCA is evaluated as the interference coefficient \(h\) increases. The proposed scheme maintains stable classification accuracy for the selected importance levels, showing its robustness against inter-symbol interference. In contrast, the tag-based baseline suffers from a noticeable degradation, especially for the medium-importance level. This degradation occurs because errors in the tag may cause the receiver to select an incorrect protection level. These results show that embedding the protection structure directly into the codebook is more reliable than using an explicit tag, particularly in short-blocklength and interference-dominated VLC scenarios.

\begin{figure}[!h]
    \centering
    \includegraphics[width=0.7\linewidth]{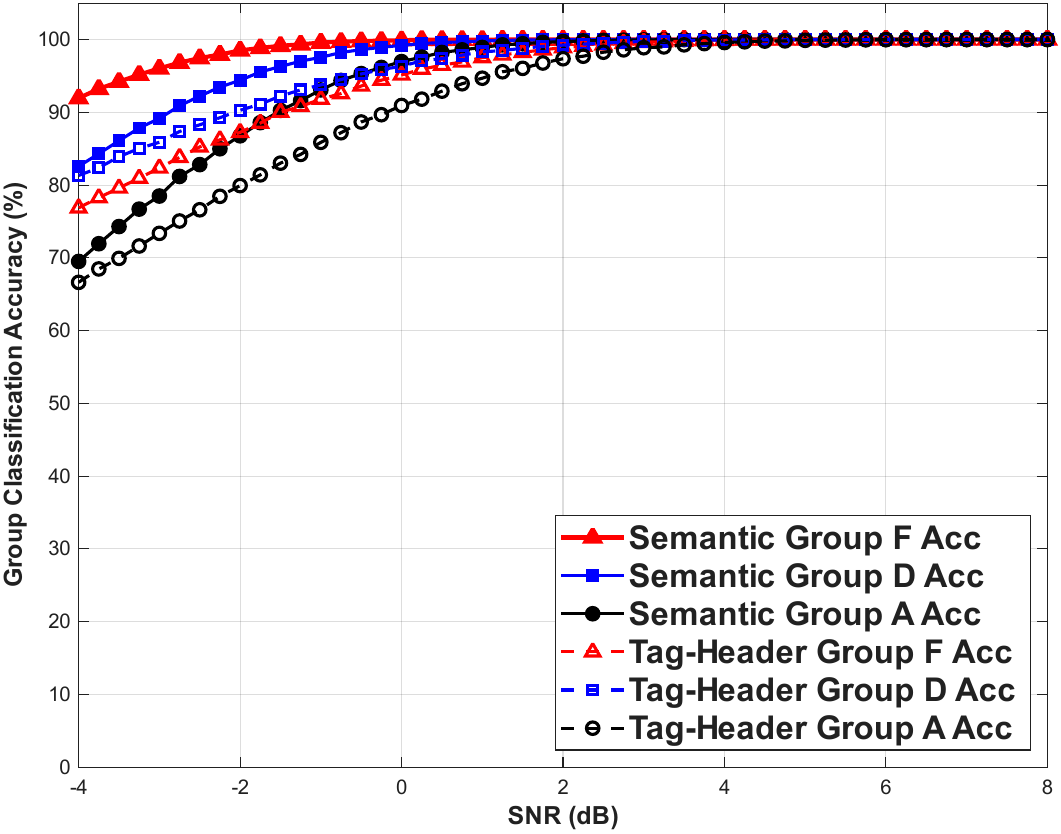}
        \vspace{-3mm}
    \caption{Group classification accuracy over the AWGN channel.} 
    \label{fig5}
    \includegraphics[width=0.7\linewidth]{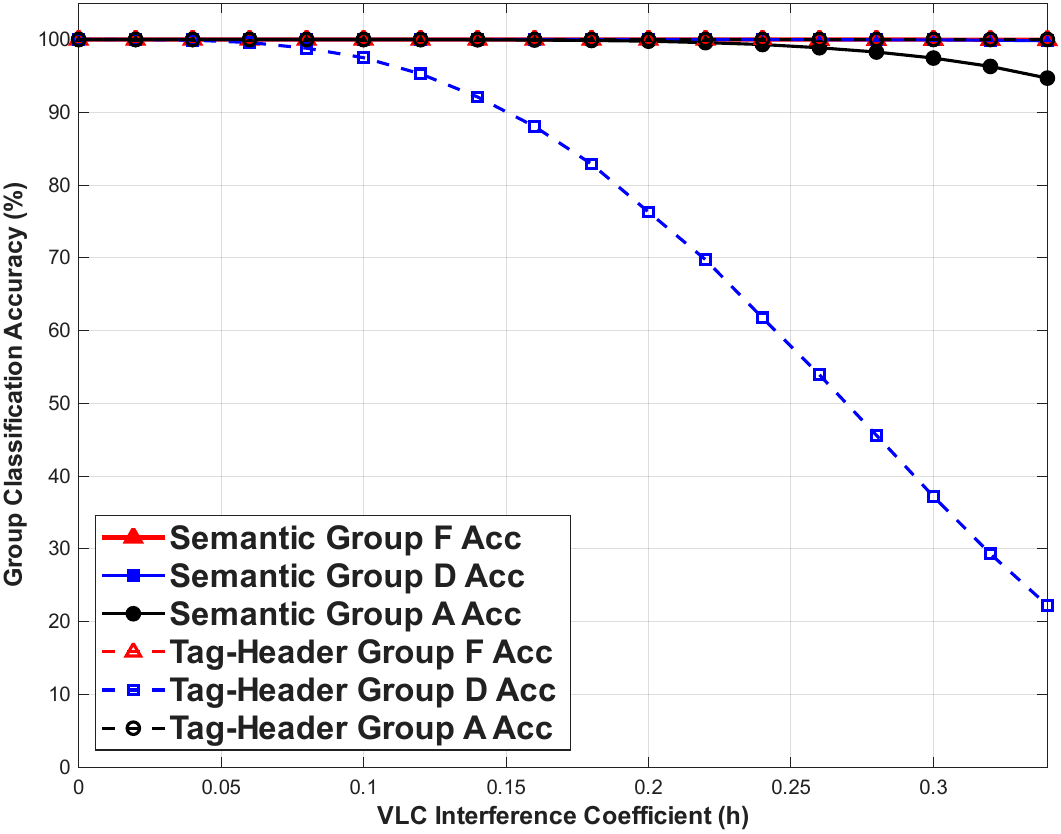}
    \caption{Group classification accuracy over the VLC-ISI channel.} 
        \vspace{-4mm}
    \label{fig6}
\end{figure}

\vspace{-4mm}
\section{Conclusion}
This paper proposed a tag-free layered construction of message-wise UEP codes with short-blocklength. By controlling the intra-level and inter-level Hamming distances of the codebook, the proposed scheme provides different error-correction capabilities for messages with different importance levels while enabling reliable importance-level classification at the receiver. 
We derived sufficient conditions for correct group classification under bounded errors and showed that groups with correction radii not smaller than the error weight cannot be falsely selected when the error weight exceeds the correction capability of the transmitted group. 
Simulation results over AWGN and VLC-ISI channels demonstrated that the proposed scheme achieves lower BER and higher group classification accuracy than the tag-based BCH baseline, particularly in interference-dominated VLC scenarios. 
Future work will focus on efficient codebook construction and low-complexity decoding algorithms for practical implementation.
    \vspace{-3mm}
\section*{Acknowledgment}
This work was partly supported by the supported by JSPS KAKENHI Grant Number 25K00371 and Telecommunications Advancement Foundation, Japan.

\bibliographystyle{unsrt}
\bibliography{sample}

@INPROCEEDINGS{shi2025,
  author={Shi, Tianhao and Lu, Shan and Yamazato, Takaya and Kuo, Yu-Hao and Ueng, Yeong-Luh},
  booktitle={2025 30th Asia-Pacific Conference on Communications (APCC)}, 
  title={Robust Detection of Overlapping LED Spots in Dense VLC via Pilot-Aided Geometric Recognition}, 
  year={2025},
  volume={},
  number={},
  pages={1-6},
  doi={10.23919/APCC64555.2025.11279960}}

@INPROCEEDINGS{Asaoka2025,
  author={Asaoka, Shin and Lu, Shan and Tang, Zhengqiang and Yamazato, Takaya},
  booktitle={2025 IEICE-CS International Conference on Emerging Technologies for Communications (ICETCCOM)}, 
  title={Theoretical Performance Analysis for Image Sensor Communication using a Propeller LED Transmitter}, 
  year={2025},
  volume={},
  number={},
  pages={1-4},
  doi={10.23919/ICETCCOM67395.2025.11416822}}

@ARTICLE{Asaoka2026,
  author={Asaoka, Shin and Lu, Shan and Tang, Zhengqiang and Yamazato, Takaya},
  journal={IEEE Communications Letters}, 
  title={ISI Modeling and BER Performance for Rotating Light-Trail Image Sensor Communication}, 
  year={2026},
  volume={},
  number={},
  pages={1-1},
  doi={10.1109/LCOMM.2026.3691874}}

@ARTICLE{Hagenauer1988,
  author={Hagenauer, J.},
  journal={IEEE Transactions on Communications}, 
  title={Rate-compatible punctured convolutional codes (RCPC codes) and their applications}, 
  year={1988},
  volume={36},
  number={4},
  pages={389-400},
  doi={10.1109/26.2763}}

@ARTICLE{Hou2016,
  author={Hou, Wei and Lu, Shan and Cheng, Jun},
  journal={IET Communications}, 
  title={Rate-compatible spatially coupled LDPC code ensembles based on repeat-accumulate extensions}, 
  year={2016},
  volume={10},
  number={17},
  pages={2422-2426},
  doi={10.1109/LCOMM.2015.2497242}}

@INPROCEEDINGS{Hou2014,
  author={Hou, Wei and Lu, Shan and Cheng, Jun},
  booktitle={2014 8th International Symposium on Turbo Codes and Iterative Information Processing (ISTC)}, 
  title={Rate-compatible spatially coupled LDPC codes via repeat-accumulation extension}, 
  year={2014},
  volume={},
  number={},
  pages={87-91},
  doi={10.1109/ISTC.2014.6955091}}

@INPROCEEDINGS{10571004,
  author={Fayaz, Nargis and Shreshtha, Aman and Sarangi, Smruti and Mallik, Ranjan K. and Lall, Brejesh},
  booktitle={2024 IEEE Wireless Communications and Networking Conference (WCNC)}, 
  title={Semantic-Aided Image Transmission System with Unequal Error Protection for Next-Generation Communication Networks}, 
  year={2024},
  volume={},
  number={},
  pages={01-06},
  doi={10.1109/WCNC57260.2024.10571004}}

@ARTICLE{mark-UEP,
  author={Shkel, Yanina Y. and Tan, Vincent Y. F. and Draper, Stark C.},
  journal={IEEE Transactions on Information Theory}, 
  title={Unequal Message Protection: Asymptotic and Non-Asymptotic Tradeoffs}, 
  year={2015},
  volume={61},
  number={10},
  pages={5396-5416},
  doi={10.1109/TIT.2015.2462846}}

@ARTICLE{9679803,
  author={Luo, Xuewen and Chen, Hsiao-Hwa and Guo, Qing},
  journal={IEEE Wireless Communications}, 
  title={Semantic Communications: Overview, Open Issues, and Future Research Directions}, 
  year={2022},
  volume={29},
  number={1},
  pages={210-219},
  doi={10.1109/MWC.101.2100269}}

@ARTICLE{5715839,
  author={Gong, Chen and Yue, Guosen and Wang, Xiaodong},
  journal={IEEE Transactions on Communications}, 
  title={Message-Wise Unequal Error Protection Based on Low-Density Parity-Check Codes}, 
  year={2011},
  volume={59},
  number={4},
  pages={1019-1030},
  doi={10.1109/TCOMM.2011.020411.090611}}

@ARTICLE{5319734,
  author={Borade, Shashi and Nakiboğlu, Bariş and Zheng, Lizhong},
  journal={IEEE Transactions on Information Theory}, 
  title={Unequal Error Protection: An Information-Theoretic Perspective}, 
  year={2009},
  volume={55},
  number={12},
  pages={5511-5539},
  doi={10.1109/TIT.2009.2032819}}

@ARTICLE{9525108,
  author={Ninkovic, Vukan and Vukobratovic, Dejan and Häger, Christian and Wymeersch, Henk and Graell i Amat, Alexandre},
  journal={IEEE Communications Letters}, 
  title={Autoencoder-Based Unequal Error Protection Codes}, 
  year={2021},
  volume={25},
  number={11},
  pages={3575-3579},
  doi={10.1109/LCOMM.2021.3108845}}

@ARTICLE{6978587,
  author={Namba, Kazuteru and Lombardi, Fabrizio},
  journal={IEEE Transactions on Computers}, 
  title={Parallel Decodable Two-Level Unequal Burst Error Correcting Codes}, 
  year={2015},
  volume={64},
  number={10},
  pages={2902-2911},
  doi={10.1109/TC.2014.2378290}}

@INPROCEEDINGS{699346,
  author={Buch, G. and Burkert, F.},
  booktitle={MELECON '98. 9th Mediterranean Electrotechnical Conference. Proceedings (Cat. No.98CH36056)}, 
  title={Non-linear codes and concatenated codes for unequal error protection}, 
  year={1998},
  volume={2},
  number={},
  pages={851-855 vol.2},
  doi={10.1109/MELCON.1998.699346}}

@INPROCEEDINGS{7848728,
  author={Barmada, Bashar and Rehman, Saeed},
  booktitle={2016 IEEE Region 10 Conference (TENCON)}, 
  title={From source coding to MIMO - a multi-level unequal error protection}, 
  year={2016},
  volume={},
  number={},
  pages={3597-3600},
  doi={10.1109/TENCON.2016.7848728}}

@ARTICLE{7334457,
  author={Tang, Hoyoung and Park, Jongsun},
  journal={IEEE Transactions on Very Large Scale Integration (VLSI) Systems}, 
  title={Unequal-Error-Protection Error Correction Codes for the Embedded Memories in Digital Signal Processors}, 
  year={2016},
  volume={24},
  number={6},
  pages={2397-2401},
  doi={10.1109/TVLSI.2015.2497368}}

@ARTICLE{yamazato,
  author={Yamazato, Takaya and Takai, Isamu and Okada, Hiraku and Fujii, Toshiaki and Yendo, Tomohiro and Arai, Shintaro and Andoh, Michinori and Harada, Tomohisa and Yasutomi, Keita and Kagawa, Keiichiro and Kawahito, Shoji},
  journal={IEEE Communications Magazine}, 
  title={Image-sensor-based visible light communication for automotive applications}, 
  year={2014},
  volume={52},
  number={7},
  pages={88-97},
  doi={10.1109/MCOM.2014.6852088}}

\end{document}